\documentclass[]{llncs}

\usepackage{rotating}
\usepackage{array}
\usepackage{amsmath,algorithm,algorithmic}
\usepackage{amssymb}
\usepackage{graphicx}
\usepackage{floatfig}
\usepackage{color}
\usepackage{type1cm}
\usepackage{fullpage}
\usepackage{subfigure}
\usepackage{wrapfig}

\begin{document}

\title{Distributed Planarization and Local Routing Strategies in Sensor Networks}
\author{Florian Huc$^1$ \and Aubin Jarry$^2$ \and Pierre Leone$^2$ \and Jose Rolim$^2$\\ \ \\
$^1$ LPD, \'Ecole Polytechnique F\'ed\'erale de Lausanne (EPFL)\\
$^2$ Computer Science Department, University of Geneva\\
}
\authorrunning{Aubin Jarry and Florian Huc and Pierre Leone and Jose Rolim}

\institute{Corresponding author: pierre.leone@unige.ch}
\date{}


\maketitle

\begin{abstract}
We present an algorithm which computes a planar 2-spanner from an Unit Disk Graph when the node density is sufficient. The communication complexity in terms of number of node's identifier sent by the algorithm is $6n$, while the computational complexity is $O(n\Delta)$, with $\Delta$ the maximum degree of the communication graph. Furthermore, we present a simple and efficient routing algorithm dedicated to the computed graph.

Last but not least, using traditional Euclidean coordinates, our algorithm needs the broadcast of as few as $3n$ node's identifiers. Under the hypothesis of sufficient node density, no broadcast at all is needed, reducing the previous best known complexity of an algorithm to compute a planar spanner of an Unit Disk Graph which was of $5n$ broadcasts.
\end{abstract}

\section{Introduction}
In many problems on networks, among which the problem of message routing  \cite{gfg,KK00,KWZ08}, it is useful to know a planar subgraph of the communication graph. Although, all planar subgraphs are not equally interesting: an usual requirement is that the length of a path between two nodes is not too much longer in the planar subgraph than in the original graph. 

In this paper, we propose a {\it{distributed and simple way to compute such a planar subgraph of a unit disk graph}} when the nodes of the communication graph are localized using the {\it{virtual raw anchors coordinate system}} \cite{HJLR10}, instead of the stronger hypothesis of having the nodes localized in a classical 2D coordinate system.

We further propose a {\it{simple, efficient and light routing algorithm}} that is dedicated to the constructed graph. This last contribution is related to the conjecture (partially solved in \cite{Dhandapani08} and fully solved \cite{DBLP:journals/dcg/LeightonM10,DBLP:conf/fct/GhoshS09}) that a 3-connected graph accepts an embedding such that the greedy routing algorithm\footnote{Given a distance among the nodes, for instance the Euclidean distance, we call the {\it{greedy routing algorithm}}, the algorithm which consists for a source $x$ to forward a message to the node that is the closest to the destination $y$ in that distance} is guaranteed to deliver any message. Indeed, we propose a routing algorithm ({\it{similar}} to the greedy routing), which ensures message delivery in the unit disk graph induced by any set of nodes, under some connectivity assumptions.

\subsection{Related work}
\setcounter{subsubsection}{0}

\subsubsection{Planar graph and poset dimension \cite{schnyder}}\label{schnyder}
it is proved that a graph $G=(V,E)$ is planar if and only if it has order-dimension at most three, where a graph has dimension $d$ if and only if there exists a sequence $<_1,\dots,<_d$ of total orders on $V$ whose intersection is empty, and such that for each edge $(xy)\in E$ and for each $z \in V\setminus \{x,y\}$, there is at least one order $<_j$ in the sequence such that $x<_j z$ and $y<_j z$.

It means that any three total orders whose intersection is empty induce a planar graph, this graph being the subgraph of the complete graph obtained by keeping only the edges that satisfy the second condition. We will refer to this graph as the Schnyder's graph of the three total order, and we note it $G^{Schnyder}_{<_1,<_2,<_3}$ or $G^{\cal S}$ for short.

\subsubsection{Planar spanner}
much work have been dedicated to the construction of planar subgraphs. One of the first planarization technique is the use of Gabriel graphs, but, if connectivity is preserved, an edge may be replaced by a path of unbounded length \cite{Gabriel}, whereas, as we mentioned previously, we want to avoid this. The following well known definition catches the type of subgraphs we are interested in: given a graph $G$ and a subgraph $G'$, we say that $G'$ is a {\it{$k$-spanner}} if, for all pairs of nodes $x,y$, a shortest path from a node $x$ to another node $y$ in $G'$, is at most $k$ times longer than a shortest path between these two nodes in $G$; the factor $k$ is called the stretch. If $G$ is the complete graph, we further say that $G'$ is a geometric $k$-spanner.

The Delaunay's triangulation of a set of vertices $V$ is a planar geometric spanner. Its stretch factor is upper bounded by $1.998$ \cite{Xia}, and lower bounded strictly by $\pi/2$ \cite{Bose2011}, the exact stretch being unknown. In \cite{Calinescu2002}, the authors efficiently construct a planar 2.5-spanner of UDG that contains the edges of length 1 of the Delaunay's triangulation of a set of nodes $V$. The complexity of this construction is improved in \cite{5Broadcast}, in which an algorithm needing 5 broadcast is proposed. An other construction of spanner of Unit Disk Graph is proposed in \cite{DBLP:conf/wads/BoseCCSX07} with stretch $>2$. Other spanners exist, in particular a way to construct a 2-spanner from a complete graph is proposed in~\cite{Chew89}. Interestingly, it is shown in~\cite{review} that three different constructions lead to the same planar geometric 2-spanner. These three constructions are the half-$\theta_6$-graphs, the triangular-distance Delaunay triangulation (TD Delaunay graphs) and the geodesic embeddings. In \cite{review2}, the authors further propose a planar spanner with bounded degree. We refer the interested reader to the recent survey of Bose and Smid \cite{survey}.


\subsubsection{Greedy embedding}
When one consider building a spanner, one usually does not focus on preserving easy routing properties such as greedy routing. It means that even if the greedy routing algorithm delivers a message to the destination in the original communication graph, it has no guarantee to succeed in an usual spanner. Preserving such a property would be of great interest.

This problem is related to the following conjecture~\cite{DBLP:conf/algosensors/PapadimitriouR04}: given a 3-connected graph, does an embedding exist such that the greedy routing algorithm is guaranteed to deliver any message ?

In \cite{Dhandapani08}, it is proved that the conjecture is true if the graph is a plane triangulation by using Schnyder's characterization of planar graphs \cite{schnyder}. Later, the conjecture was proved for every 3-connected graphs in \cite{DBLP:journals/dcg/LeightonM10,DBLP:conf/fct/GhoshS09}.

\subsection{Summary of results}
\setcounter{subsubsection}{0}

\subsubsection{VRAC}
For all the results presented in this paper, we assume that the nodes are localized using the virtual raw anchors coordinate system (VRAC, \cite{HJLR10}), or a simple variant of it. 
Supposing that the nodes are localized in these coordinate systems, is a strictly weaker hypothesis than the hypothesis that the nodes are localized in a traditional 2D coordinate system. Indeed, if the nodes are localized in a traditional 2D coordinate system, it is possible to compute their coordinates in the virtual raw anchors coordinate systems, while the converse is impossible. Furthermore, this coordinate system is expected to be easier to implement in practice.

\subsubsection{Planar subgraph and spanner}
As mentioned previously, a graph of order-dimension 3 is planar \cite{schnyder}. Hence, to planarize a graph, it is sufficient to select edges that correspond to three total orders. This technique has already been used, and, for instance, the half-$\theta_6$ graph mentioned in the previous section can be constructed along this line. However, several issues appears. First, this technique may need important computations. Second, the three orders being total, the computation may not be feasible locally. Third, the computed planar graph is a subgraph of the complete graph, and may not be a subgraph of the communication graph. In this paper, we address these issues when the communication graph is a Unit Disk Graph. To do so, we propose three total orders based on the VRAC coordinates using which we can construct a {\emph{planar graph}} such that if the node density is high enough,
\begin{enumerate}
\item it needs only comparison (no other operations of any type), 
\item it is a 2-spanner 
\item constructing it requires to broadcast at most $6n$ nodes identifiers.
\end{enumerate} 

In particular, our result improves the result of \cite{Calinescu2002} by construction a spanner with stretch factor 2 versus 2.5. Plus, using the VRAC coordinates, our algorithm induces the broadcasts of at most $6n$ node's identifiers (excluding the one needed for the neighborhood discovery), and has computational complexity $O(n\Delta)$, for $\Delta$ the maximum degree of $G$. Furthermore, using traditional Euclidean coordinates, it needs the broadcast of as few as $3n$ node's identifiers that can be done in a single communication round. As when the density is high enough, the constructed graph is a planar 2-spanner, {\emph{our work answers the open problem number 22 of \cite{survey}}}.

If this work is inspired by the paper of Schnyder \cite{schnyder}, we stress out that the constructed graphs are not necessarily subgraphs of $G^{\cal S}$, the planar graph induced by the three total orders as when following Schnyder's theory.

In more details, in Section \ref{sec:3} we construct a first subgraph $\widetilde G$ from $G$ by using only each node's neighbors. If we only keep the edges of length at most $2r/\sqrt 5 \approx 0.8944 r$, the subgraph $\widetilde G$ is planar. When the node density is too small, the obtained graph may not be connected, and, to avoid this, we introduce virtual edges (edges that are not edges of the connexion graph). 
In Section \ref{greedy}, changing slightly the VRAC coordinates system, we prove that $\widetilde G\prime$ is a subgraph of the half-$\theta_6$ graph which is equals to $G^{\cal S}$. It implies that $\widetilde G\prime$ is planar, but it gives no result on its stretch. Nonetheless, we prove that it verifies 1) the length of a shortest path in $\widetilde G\prime$ is at most twice the length of a shortest path in $G$, and 2) a virtual edge corresponds to a path of two edges in $G$. When the node density increases, the virtual edges disappear, hence, $\widetilde G\prime$ is a planar 2-spanner of $G$ when the node density is high enough. All these results hold even when the constructed graph is not equal to the half-$\theta_6$ graph. Finally, using the VRAC coordinates, our algorithm needs two round of communications and induces the broadcasts of $6n$ node's identifiers on top of the one needed for the neighborhood discovery. Using Euclidean coordinates, we can reduce this to the broadcasts $3n$ node's identifiers that can be performed in a single round of communication, and no messages are exchanged at all when the density is high enough.

\subsubsection{Routing}
In \cite{Dhandapani08}, it is proved that a plane triangulation has an embedding in which the greedy routing algorithm is guaranteed to deliver any message.
In our work, we assume given the embedding in the plane, so, it means that, instead of choosing an embedding for the nodes, we look at the dual problem, that is designing a routing algorithm (as close as possible to the greedy routing algorithm) which guarantees delivery. 

\section{The model}\label{model}
\subsection{Communication model}
\setcounter{subsubsection}{0}

We consider a wireless network in which two nodes can communicate if they are at distance at most $r$, the communication radius. We can normalize the distances so that $r=1$, in which case we have Unit Disk Graphs (UDG). However, we will keep mentioning $r$, as we believe it carries useful information. The use of the UDG model for the communication links is subject to caution from a practical point of view. We quickly mention the recent paper \cite{DBLP:conf/ipps/LebharL09} that discusses how protocols that are proved valid under the UDG model can be turned to valid protocols in the more realistic SINR model. Another way of extending the results of this paper to more general communication models is to use basic properties of such models like the convexity of the region where the communication can happen \cite{DBLP:conf/podc/AvinEKLPR09}. Indeed, it seems to us that most of the arguments that we use are related to this property.

The communication graph is given by the structure $(V,E)$ where $V$ is the set of nodes and $E$, the set of edges, i.e. the set of couples of nodes that can communicate together directly. We will use virtual edges. A virtual edge is an edge between two nodes $x$ and $y$ such that $(xy)\not\in E$, but with a path from $x$ to $y$ of edges of $E$.

Finally, we do not consider the impact of interferences or collisions during wireless communication.

\subsection{Coordinate system}\label{sec:def}
\setcounter{subsubsection}{0}
We use the virtual raw anchor coordinate system \cite{HJLR10} with three anchors $A_1, A_2, A_3$. It means that each node knows its distances to the three anchors, distances which form the node coordinates. I.e. the coordinates of node $x$ is the vector $\left(d(x,A_1),d(x,A_2),d(x,A_3)\right)$. 

\begin{definition} The coordinates of a node $x$ is a vector
\begin{equation*}
\left(x_1,x_2,x_3\right)=\left(d(x,A_1),d(x,A_2),d(x,A_3)\right)
\end{equation*}
\end{definition}

Throughout the paper, we suppose that all nodes lay inside the triangle defined by the three anchors on a 2D-plane, this area is denoted $\cal A$. We use two different distances to define the coordinate system. 

In Section \ref{sec:3}, we use the Euclidean distance for the distance function $d$. Given two points $x$ and $y$, we note $|xy|$ the Euclidean distance from $x$ to $y$, throughout the paper. 

In Section \ref{greedy}, we extend the results using for the distance $d(x,A_1),d(x,A_2),d(x,A_3)$, the heights of the triangles $\widehat{A_2xA_3}$, $\widehat{A_1 x A_3}$ and,$\widehat{A_1 x A_2}$ respectively. We will note this distance $d^h(x,A_i)$ or $d_{A_i}^h(x)$ for $1 \leq i \leq 3$. We further suppose that $\widehat{A_1A_2A_3}$ is equilateral and that all nodes know the distances between the anchors: $| A_1A_2|$, $| A_1A_3|$,$| A_2A_3|$.

\section{Distributed graph planarization}\label{sec:3}
In this section, given an Unit Disk Graph $G$, we build a planar subgraph $\widetilde G$. We further extend it to $\widetilde G\prime$ by changing some of its edges by virtual edges, where a virtual edge represents a path of $G$.

Recall that in \cite{schnyder} it is proven that if we consider a graph $G=\left(V,E\right)$ and that we have three total order relations, $<_1, <_2, <_3$, on the set of nodes and
\begin{itemize}
\item the intersection of the three order relations is empty,
\item for each edge $(x,y)\in E$ and for each vertex $z\not\in\{x,y\}$ there is at least one order $<_i$ such that $x<_i z$ and $y<_i z$.
\end{itemize}
then the graph admits a planar embeddings. In this paper, we adapt this result to UDG. It leads to a simple and {\emph{localized distributed}} algorithm to planarize a communication graph of a wireless network and to a simple description of the communication graphs that accept an efficient routing algorithm.

Our aim is to define three suitable order relations by using virtual raw anchor coordinate system. The order relations have to satisfy some properties to ensure that the resulting communication graph admits a planar embedding. In this paper, we show how to {\it{locally}} compute the planar embedding by using only the distances to the anchors. We assume that there are no pair of nodes $x,y$ such that for a given $k\in\{1,2,3\}$, $d(x,A_k)=d(y,A_k)$. It does not restrict the generality since the (Lebesgue) measure of these positions is zero.

Given the nodes' coordinates, we define three total order relations, $ <_1,  <_2,  <_3$ on the set of nodes $V$ in the following way:

\begin{definition} For $k \in \{1,2,3\}$, nodes $x$ and $y$ with coordinates $(x_1,x_2,x_3)$ and $(y_1,y_2,y_3)$ satisfy the relation $x~ <_k~ y$, if and only if $x_k<y_k$.
\end{definition}

\begin{lemma}\label{intervide} Given that the three anchors are not aligned, we consider the set of nodes that are inside $\cal A$, see Figure \ref{fig:notempty}. If we consider the restriction of the order relations $<_k$ on $\cal A\times A$ denoted $<_k|_{\cal A}$ then their intersection is empty.
\begin{equation}\label{eq:emptyinter}
\bigcap_{k=1}^3 <_k|_{\cal A} = \emptyset.
\end{equation}
\end{lemma}
\begin{proof} To prove that the intersection is empty is equivalent to prove that given any point $x$ that belongs to the convex hull of the three anchors the triangular area ${\cal A}$ is covered by the three circles centered on the anchors and passing through $x$. Indeed, if the intersection is not empty there is a point $y\in {\cal A}$ that belongs outside of the three circles (and reciprocally) , i.e. $x~<_k~ y$, for $k\in \{1,2,3\}$. 

Because the area ${\cal A}$ is the union of the three triangles    $\widehat{A_1 x A_3}$, $\widehat{A_1 x A_2}$ and $\widehat{A_2 x A_3}$, see Figure \ref{fig:nocrossingb},  it is sufficient to show that the three triangles are covered by the circles. We consider $\widehat{A_1 x A_3}$ particularly and the proof extend to the others triangles. We decompose the triangle $\widehat{A_1 x A_3}$ into two sub-triangles $\widehat{A_1x x\prime}$ and $\widehat{A_3x x\prime}$, where $x\prime$ is such that the line $xx\prime$ crosses the line $A_1A_3$ perpendicularly.  Because the length of the segment $A_1x$ is larger than the length of the segment $A_1x\prime$ the sub-triangle $\widehat{A_1xx\prime}$ is covered by the circle centered in $A_1$ and passing through $x$. The same argument apply to the sub-triangle $\widehat{A_3xx\prime}$ and this concludes the proof.
\end{proof}

\begin{figure}
\begin{center}
\scalebox{0.2}{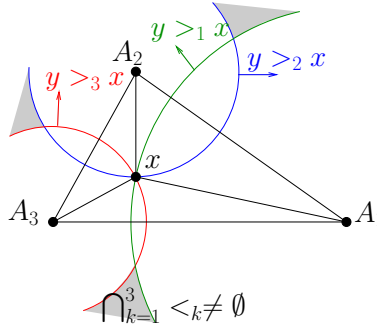}\caption{If we do not restrict ourselves to the region $\cal A$, $\bigcap_{k=1}^3 <_k \neq \emptyset$. $\bigcap_{k=1}^3$ is represented in gray}\label{fig:notempty}
\end{center}\end{figure}

\begin{remark}
Notice that if we do not assume that the nodes belong to the area $\cal A$ then, the intersection $(\ref{eq:emptyinter})$ may not be empty. Indeed there are point $y$ whose the distances to the three anchors are larger than the distances of $x$ to the three anchors, i.e. $x<_k y,~\forall k=1,2,3$, see Figure~\ref{fig:notempty}.
\end{remark}

\begin{definition} We define the three binary relations $\widetilde <_1, \widetilde <_2, \widetilde  <_3$ by $\forall x,y \in V,~ k=1,2,3, ~ x\widetilde <_k y$ $\Longleftrightarrow~ x~ <_k~ y\text{ and } y~ <_j~ x\text{ for } j\not = k.$
\end{definition}

From Lemma \ref{intervide}, we deduce that the graph $G^{Schnyder}_{\widetilde<_1,\widetilde<_2,\widetilde<_3}$ induced by these three total orders is planar. However, as we mentioned in the introduction, there are some major issues: 1) $G^{Schnyder}_{\widetilde<_1,\widetilde<_2,\widetilde<_3}$ may not be a subgraph of an UDG, and 2) $G^{Schnyder}_{\widetilde<_1,\widetilde<_2,\widetilde<_3}$ can not be computed locally.

We denote $min_k$ the minimal $z$ with respect to the order relation $<_k$. The next lemma, gives a {\emph{local}} condition to ensure planarity.

\begin{lemma}\label{lemma:planarcond} Given an UDG $G=(V,E)$, if $\forall (x,y)\in E$ and $\forall z\in V\setminus\{x,y\}$ with $\max\{|xz|, |yz|\} < \sqrt{5}r/2$, there exists $k \in \{1,2,3\}$ such that $x<_k z$ and $y<_k z$ then the graph is planar.
\end{lemma}
\begin{proof} 
The condition that $\max\{|xz|, |yz\} < \sqrt{5}r/2$ is particular to UDG. Indeed, for an edge $(u,v)\in E$ with either $u$ or $v$ at distance larger than $\sqrt{5}r/2$ from $x$ and $y$, the two edges $(xy)$ and $(uv)$ (whose lengths are bounded by $r$) cannot intersect. This condition limits the set of nodes that potentially can be linked to an edge intersecting $(xy)$ and ensures that the verification can be done locally.

We now consider two edges $(xy)$ and $(uv)$, with both $u$ and $v$ at distance at most $\sqrt{5}r/2$ from $x$ and $y$. By assumption, there exist $k_1, k_2, k_3, k_4 \in \{1,2,3\}$ such that
\begin{equation*}
u,v ~<_{k_1}~x,~u,v~<_{k_2}~y,~x,y~<_{k_3}~u,~x,y~<_{k_4}~v.
\end{equation*}

It is clear that $k_1\not = k_3, k_4$ and $k_2\not = k_3, k_4$ and we can assume that $k_1=k_2$ and then $u,v ~<_{k_1}~\text{min(}x,y\text{)}$. Indeed, if $k_3\not = k_4$ we have $k_1=k_2$ because $k_i=1,2,3$. If $k_3= k_4$ we apply the same argument to $u,v$ instead of $x,y$.

We conclude that $(uv)$ do not cross $(xy)$ because each point of $(uv)$ are $<_{k_1}$ smaller than $x$ and $y$, see Figure \ref{fig:nocrossinga}. 
\end{proof}

\begin{figure}
\begin{center}
\subfigure[$u,v~<_{k_1}~\text{min(} x,y)$ implies that $(uv)$ and $(xy)$ do not intersect.]{\scalebox{0.14}{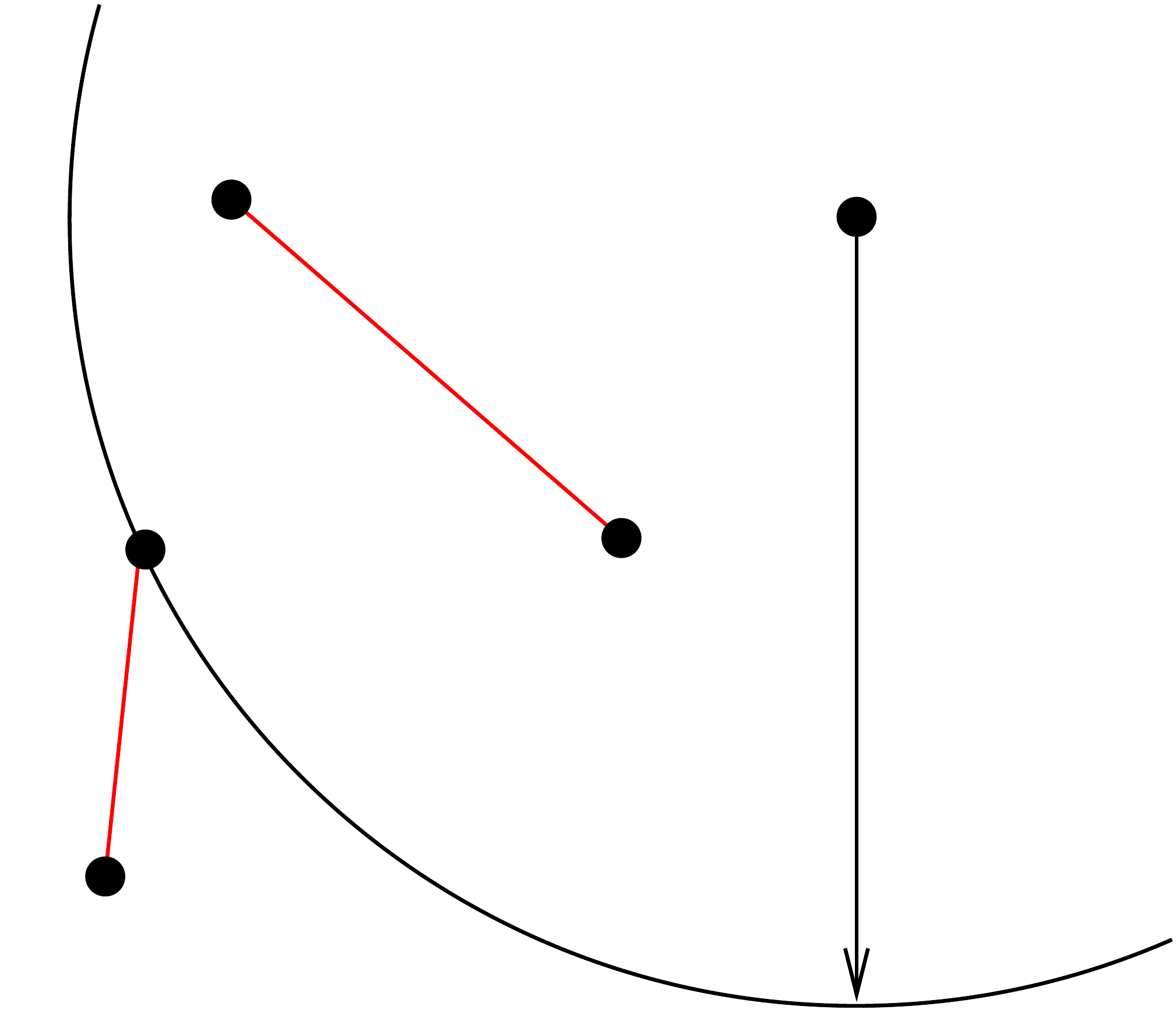}\label{fig:nocrossinga}} \hspace{2cm}
\subfigure[$\bigcap_{k=1}^3 <_k|_{\cal A} = \emptyset.$]{\scalebox{0.21}{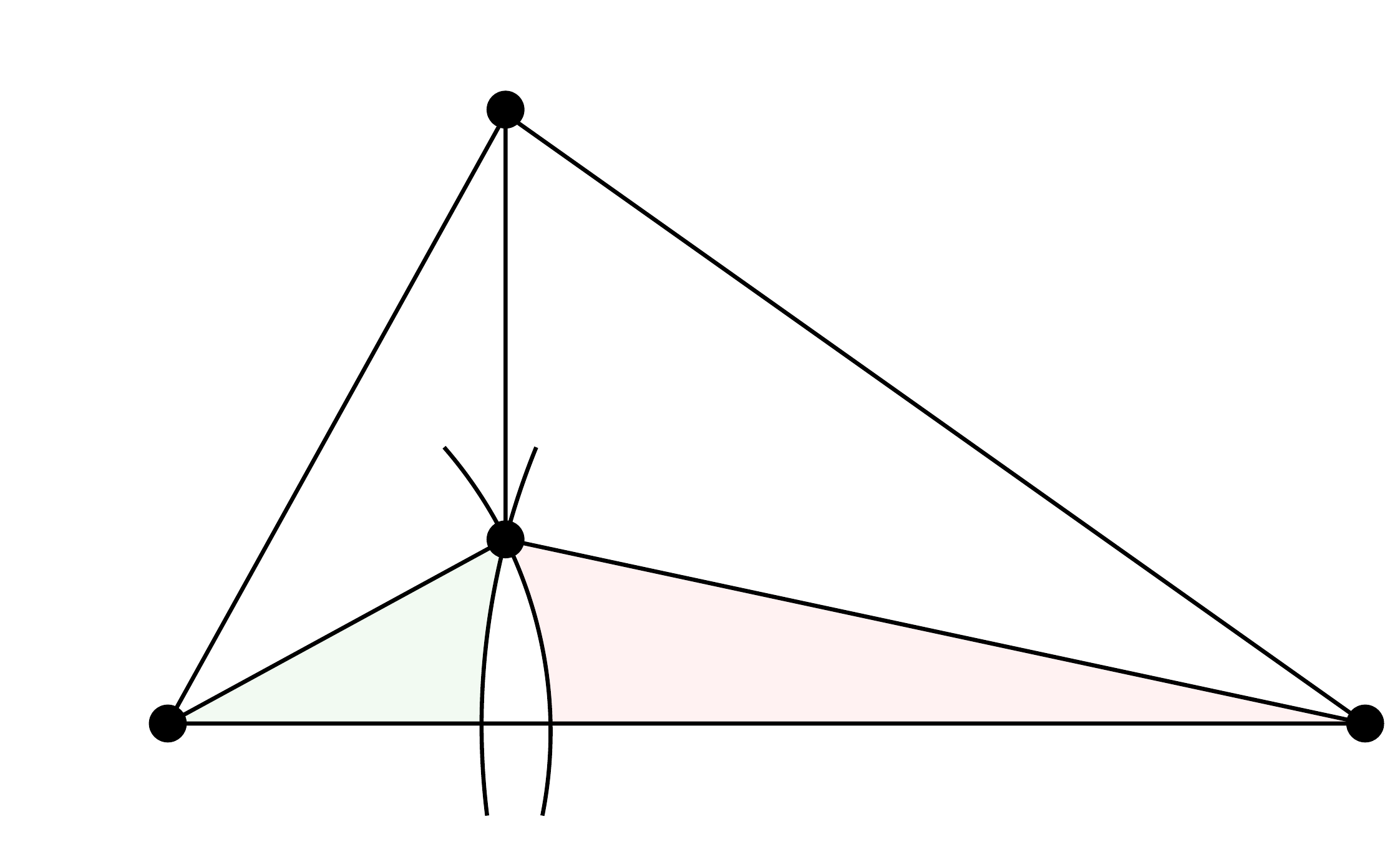}\label{fig:nocrossingb}}
\caption{}
\end{center}\end{figure}

\begin{lemma}\label{lemmarcond}  Given a graph $G=\left(V,E\right)$ and three anchors $A_1, A_2, A_3$. We define the subgraph $\widetilde G=\left(V, \widetilde E\right)$ of $G$ by $\forall x,y\in V, (xy)\in\widetilde E \Longleftrightarrow  (xy)\in E $ and $ \exists k \in \{1,2,3\}$ such that $y=min_k\{z\mid x~\widetilde <_k~ z\}$ or $ x=min_k\{z\mid y~\widetilde <_k~ z\}.$
If the lengths of all the edges in the resulting graph $\widetilde G=\left(\widetilde V, \widetilde E\right)$ are smaller than $2r/\sqrt{5}\approx 0.8944$ then the graph $\widetilde G$ is planar.
\end{lemma}
\begin{proof}
Let $(xy)$ be an edge of $\widetilde G$. Since by hypothesis, all edges of $\widetilde G$ are of length at most $2r/\sqrt{5}$, it can be seen as a subgraph of a Disk Graph with radius $2r/\sqrt{5}$. To apply Lemma \ref{lemma:planarcond}, it is sufficient to check that $\forall z\not\in \{x,y\}$ with $\max(|xz|,|yz|\leq \sqrt{5}/2 * 2r/ \sqrt{5} = r$, there exists $k\in \{1,2,3\}$ such that $x,y~<_k~z$. 
Only, if such a $z$ exist, in our construction of $\widetilde G$, we would not have the edge $(xy)$, but instead, we would have an edge $(xz)$, a contradiction.

\end{proof}

\begin{remark}
The selection procedure of the edges naturally induces an orientation. Hence the obtained graphs may be seen as digraphs. We will use this remark in Section~\ref{greedy}.
\end{remark}

$\widetilde G$ is not a subgraph of $G^{Schnyder}_{\widetilde<_1,\widetilde<_2,\widetilde<_3}$, however, Lemma \ref{lemmarcond} provides a sufficient condition ensuring that the subgraph $\widetilde G$ of the communication graph $G$ is a 
planar graph, see Figure \ref{fig:planargraphs}. The advantage of $\widetilde G$ is that it needs each node to know only its neighbors. It means that it is sufficient that each node broadcasts its identifier and its VRAC coordinates, so that its neighbors know them. It induces a communication complexity of $O(\log(n))$ bits. When the density is high enough, the proposed condition is sufficient. 
However, when the density is low, we observe that there are situations in which considering only the edges of length at most $\frac{2r}
{\sqrt{5}}$ leads to disconnect the graph. To avoid this, a solution is to reconstruct $G^{Schnyder}_{\widetilde<_1,\widetilde<_2,\widetilde<_3}$, but two questions arise: can we  still do it locally ? If it is not a subgraph of the communication graph, how can we detect the missing edges ? For this, the solution that we propose is to introduce {\bf virtual links}. If we are in the situation where $x~\widetilde<_k~ y$ and 
$x~\widetilde<_k ~z$, $z~<_k~y$ but $z$ is out of range of communication of $x$, from Schnyder's theory, $x$ would rather be connected to $z$ than $y$. Only, due to the limited range of $x$, it does not occur. Then, the edge $(xy)$ can potentially cross an edge from $z$ if $|zy|<r$. To avoid this, we replace $(xy)$ by a virtual edge $(xz)$. For this, the node $y$ that knows its neighborhood informs $x$ and a virtual edge between $x$ and $z$ (through $y$) replaces the edges $(xy)$. In turn, $z$ also checks if it is in the same situation as $y$. Ultimately, we would like that the computed graph $\tilde G \prime=(V,\tilde E \prime)$ is a subgraph of the graph induced by the three total orders.

The algorithm to compute the $\tilde G \prime$ goes as follow:
\begin{itemize}
\item (As in Lemma \ref{lemmarcond}) Each node $x$ knows its neighboring nodes and compute the nodes $y_k,~k=1,2,3$ such that $y_k=\min\{z\mid x~\widetilde <_k ~z\}$.
\item (Virtual edge) Each node $x$ checks with its neighboring nodes $y_k,~k=1,2,3$ that there does not exist a node $y_k\prime$ in its second neighborhood such that $y_k\prime~<_k y_k$, $x~\widetilde <_k~ y\prime$ ($d(y\prime,y)<r$) and $y_k\prime$ is out of the communication range of $x$.
\begin{itemize} \item If such a node does not exist the edge $(xy)$ becomes active. 
\item If such a node exists, $y_k$ check with $y_k\prime$ that there does not exist a similar node that is out the range of communication of $y\prime$. This operation is repeated recursively until no node satisfying this property is found and a virtual edge is created between $x$ and the last node found. The original edge $(xy_k)$ is removed.
\end{itemize}
\end{itemize}

\begin{remark}
In the next section, we will see that, using the modified VRAC coordinates, a virtual edge represents a path of length 2. In this case, the recurrence is useless. However, using the original VRAC coordinates, a virtual edge can represent a longer path.
\end{remark}

In the middle of Figure \ref{fig:planargraphs} we plot a communication graph resulting from the selection described in Lemma \ref{lemmarcond} without restricting the lengths of the edges to be smaller than $\frac{2r}{\sqrt{5}}$. We observe that two edges cross. In the right of the Figure \ref{fig:planargraphs} the virtual links mechanism is used. We observe that the crossing is removed and the graph is planar. By comparing with the left side of the figure, we observe that the connectivity of the graph is better with the virtual links. However, when the node density is high enough, the selection of edges of length less than $\frac{2r}{\sqrt{5}}$ is sufficient. Furthermore, when the density increases, the number of virtual edges tends to zero (cf Section \ref{experiments}.

\begin{center}
\begin{figure}
\includegraphics[scale=0.35]{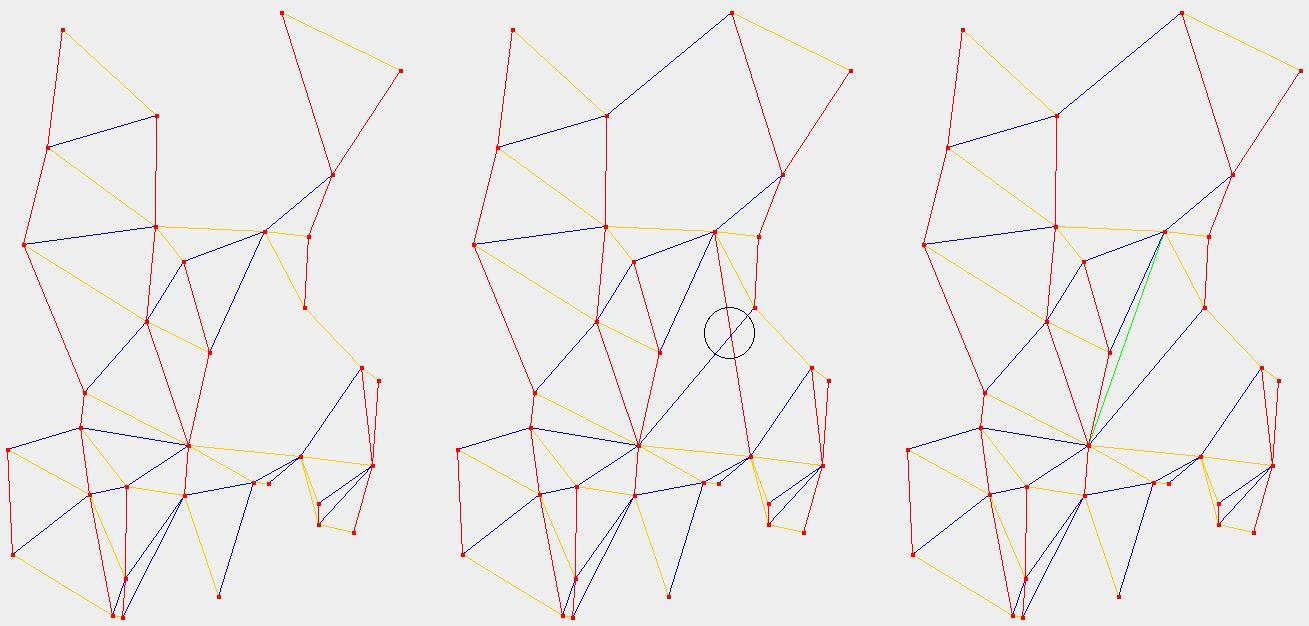}
\caption{On the left the planar graph obtained by considering only edges of length at most $\frac{2r}{\sqrt{5}}$. On the middle the graph obtained by considering all the edges, we observe two edges crossing. On the right the planar graph obtained with virtual links (in green), the crossing is removed.}\label{fig:planargraphs}
\end{figure}
\end{center}

$\tilde G \prime$ may not be a subgraph of $G$, which means that some of its edges may have length greater than $r$. For this, we can not use Lemma \ref{lemma:planarcond}. However, if we can prove that it is a subgraph of $G^{\cal S}$, we would obtain from \cite{schnyder}, that it is planar. In the next section, we slightly change the VRAC coordinate system in order to prove that it is the case. It also allows us to give guarantees on the stretch of $\widetilde G \prime$.

\section{Properties of the planar embedding}\label{greedy}
In this section we discuss a simple extension of the VRAC coordinate system. Using this new coordinate system, we prove that in $\widetilde G\prime$, the distance between two nodes is at most twice the distance in the original graph $G$. Furthermore, we show that a virtual edge $\widetilde e$ corresponds to a path of at most two edges in $G$, and that the length of such an edge is upper bounded by $2r/\sqrt{3}$.


In this section, we make the following hypothesis:
\begin{itemize}
\item There are three anchors $A_1, A_2, A_3$, the nodes belong to the convex hull $\cal A$ of the anchors and they know their distances to all three anchors.
\item $\widehat{A_1A_2A_3}$ is equilateral.
\item the nodes know the distances between the anchors ($| A_1A_2|$, $| A_1A_3|$,$| A_2A_3|$).
\end{itemize}
With respect to the first part of the paper, the two last hypothesis are new. By using the distances between the anchors, each node $x$ can compute the heights of the triangles $\widehat{A_2xA_3}$, $\widehat{A_1 x A_3}$ and, $\widehat{A_1 x A_2}$. As $\widehat{A_1A_2A_3}$ is equilateral, it is equivalent to compute the surface of these triangles or the heights and it is then easy to see that their sum is constant. We denote these values $(x_1,x_2,x_3)$, see Figure \ref{fig:VRACTriangle}. 
 One advantage of this coordinate system with respect to just using the distances to the anchors is that the sum of the three triangle areas is constant, so we can normalize the coordinates such that $x_1+x_2+x_3=1$. A reason for doing this is that because of the measurement errors on the physical location of the nodes, it is likely that the distances to the anchors do not correspond to coordinates inside a same plane. With the normalization, we project the coordinate on a same plane.

\subsection{Adapting results of Section \ref{sec:3} and further}
\setcounter{subsubsection}{0}
\subsubsection{Results of Section \ref{sec:3}}
Using the coordinates defined above, we define the order relations $<_1, <_2, <_3$ and $\widetilde<_1, \widetilde<_2, \widetilde<_3$ the same way we did in Section \ref{sec:3}.
In Section \ref{sec:3}, given a node $x$, the nodes satisfying $y \widetilde >_k x$ were outside the circle centered at $A_k$ of radius $|xA_k|$. With the new definition of the distance function $d$ (c.f. Section \ref{sec:def}), the nodes $y$ satisfying $y \widetilde >_k x$ are contained on the half plane containing $A_k$ defined by the line parallel to $(A_{k\ mod\ 3+1}A_{k\ mod\ 3+1})$ going through $x$, as illustrated in Figure \ref{fig:defg}.
\begin{figure}[htb]
\begin{center}
\subfigure[The new coordinate system: $x=(x_1, x_2, x_3)$ where $x_1,x_2$ and $x_3$ are respectively the heights of the triangles $\widehat{A_2xA_3}$, $\widehat{A_1 x A_3}$ and,$\widehat{A_1 x A_2}$.]{\scalebox{0.3}{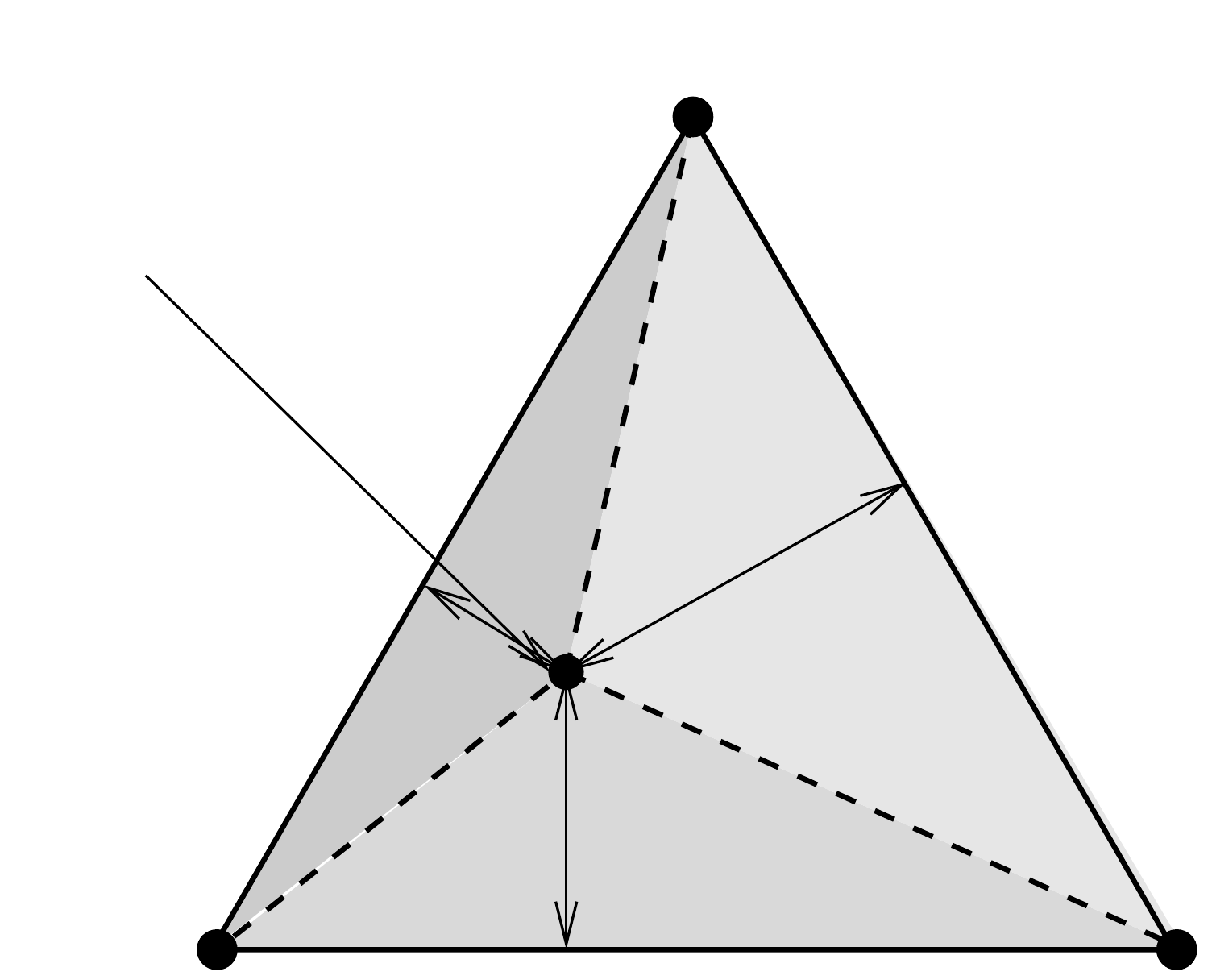}\label{fig:VRACTriangle}}\hspace{2cm}
\subfigure[$y$ with $x\widetilde<_3 y$.]{\scalebox{0.27}{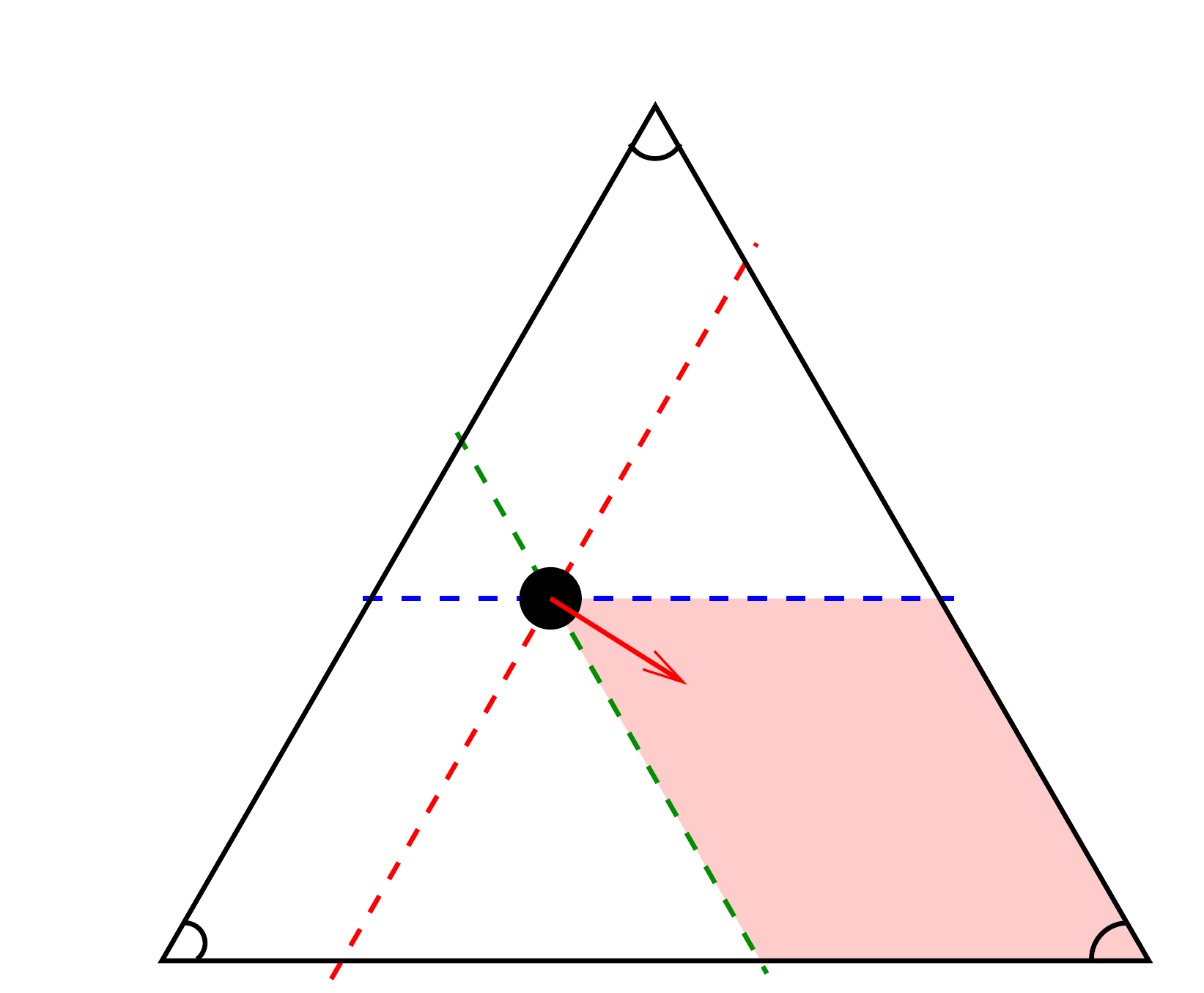}\label{fig:defg}}
\\
\subfigure[New proof of Lemma \ref{lemma:planarcond}]{\scalebox{0.27}{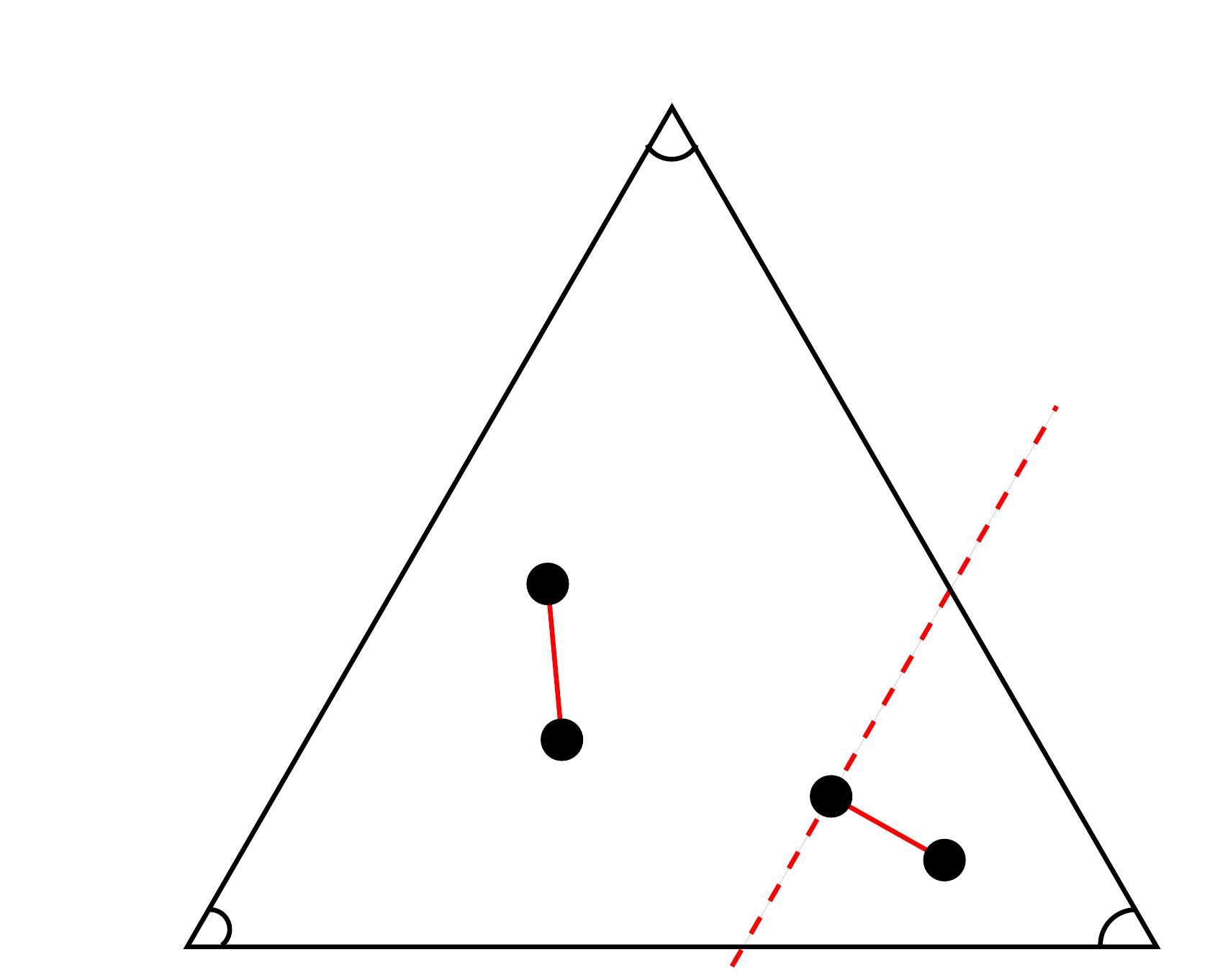}}\hspace{2cm}
\subfigure[Definition of the greedy regions]{\scalebox{0.27}{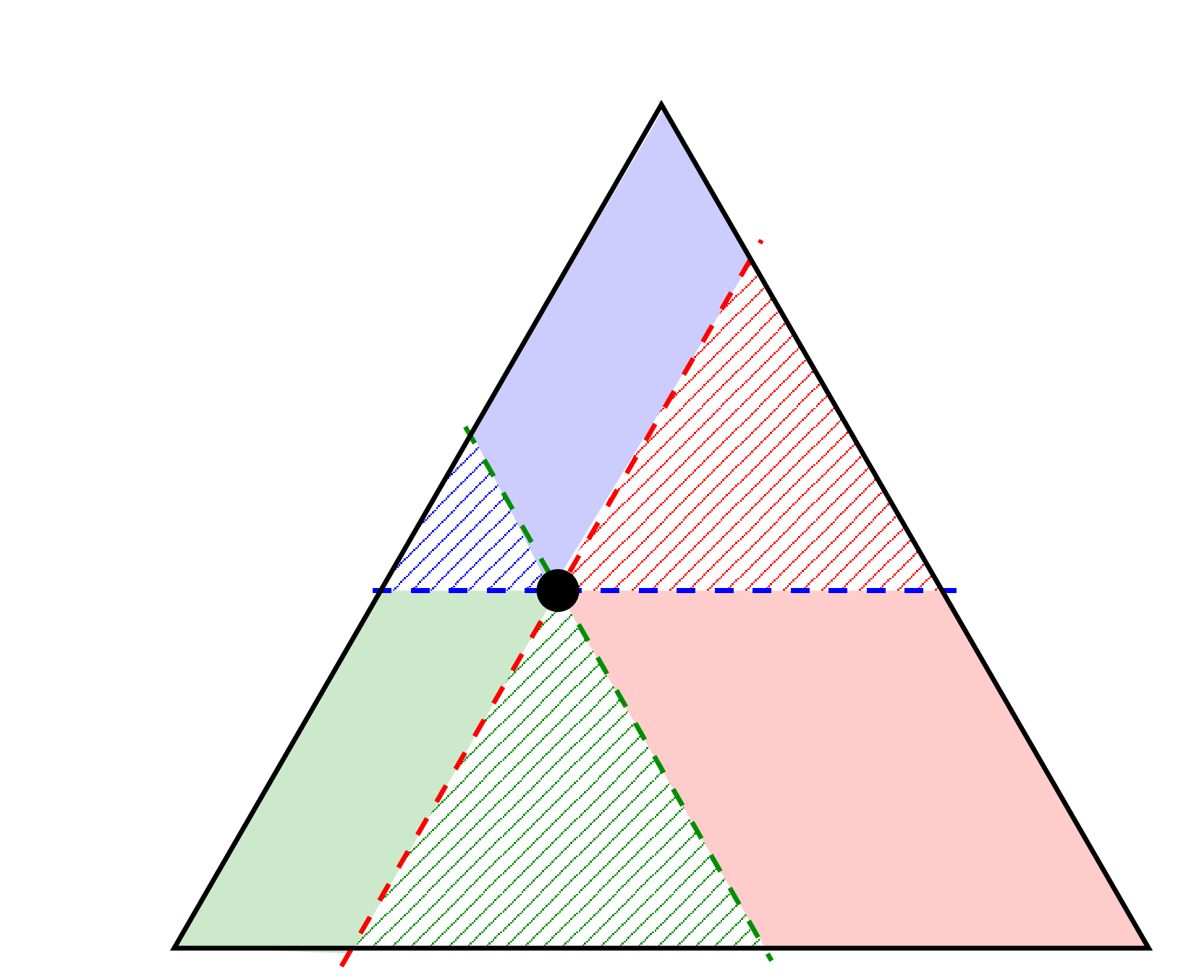}\label{fig:def2}}
\caption{Adapting proofs to the new settings.}
\end{center}\end{figure}
Using this observation, it is easy to see that the intersection of the three order relations is empty, so Lemma \ref{intervide} is still valid in this coordinate system. 
Similarly, Lemma \ref{lemma:planarcond} remains also true and we can adapt the proofs of Lemma \ref{lemmarcond}.

To summarize, {\bf all the results we have proved previously are valid with the new coordinate system}. 

\subsubsection{Connectivity results and stretch}
\ \newline
\indent{\it{Definition}} (Figure \ref{fig:def2}): Given a node $x$, we call the {\bf greedy regions} of $x$ the three regions $A_i^x=\{z\mid x~\widetilde <_i ~z\}$, for $i\in \{1,2,3\}$.

{\it{Definition}} (Figure \ref{fig:def2}): Given a node $x$, for $i\in \{1,2,3\}$, we denote the region between the two regions $A_i^x$ and $A_{i\ mod\ 3 +1}^x$ by $\bar A_i^x$.

\begin{remark}
A node $x$ has at most one outgoing edge towards a node in each of its greedy regions. It has no outgoing edge towards node not in its greedy regions, however, it may have an ingoing edge from any node.
\end{remark}

\begin{lemma}\label{edgeFactor2}
Given an edge $(x,y)\in E$, there is a path $P$ from $x$ to $y$ in $\widetilde G\prime $. The path is contained in either $\{z \in A_i^x| z \leq_i y\}$ or $\{z\in A_i^y | z\leq_i x\}$ for some $i\in \{1,2,3\}$, and it verifies: $\sum_{e \in P} | e| \leq 2|xy|$.
\end{lemma}

\begin{proof}
Given an edge $(x,y)\in E$, without loss of generality, we can suppose $y \in A_1^x$. $(x,y)\in E$. By hypothesis, $|xy| \leq r$.

We prove the lemma by induction on the length of $xy$\footnote{The length of an edge takes value in $I\hspace{-1mm}R$, however there are a finite number of edges (upper bounded by $n^2$), hence our induction will terminate.}. By our definition and Lemma \ref{lemmarcond}, $(x,y)\in \widetilde E \prime$ iff $y=min_1\{z\mid x~\widetilde <_1~ z\}$.

If $y=min_1\{z\mid x~\widetilde <_1~ z\}$, then $(x,y) \in \widetilde E \prime$ and there is a (direct) path between $x$ and $y$.

If not, there is $x'$ with $x' \widetilde <_1~ y$ and $x'=min_1\{z\mid x~\widetilde <_1~ z\}$. Notice that $x' \in \{z \in A_i^x|z \leq_i y\}$.

We now prove that $|x'y| < |xy|$.

\begin{figure}
\begin{center}
\subfigure[]{\scalebox{0.3}{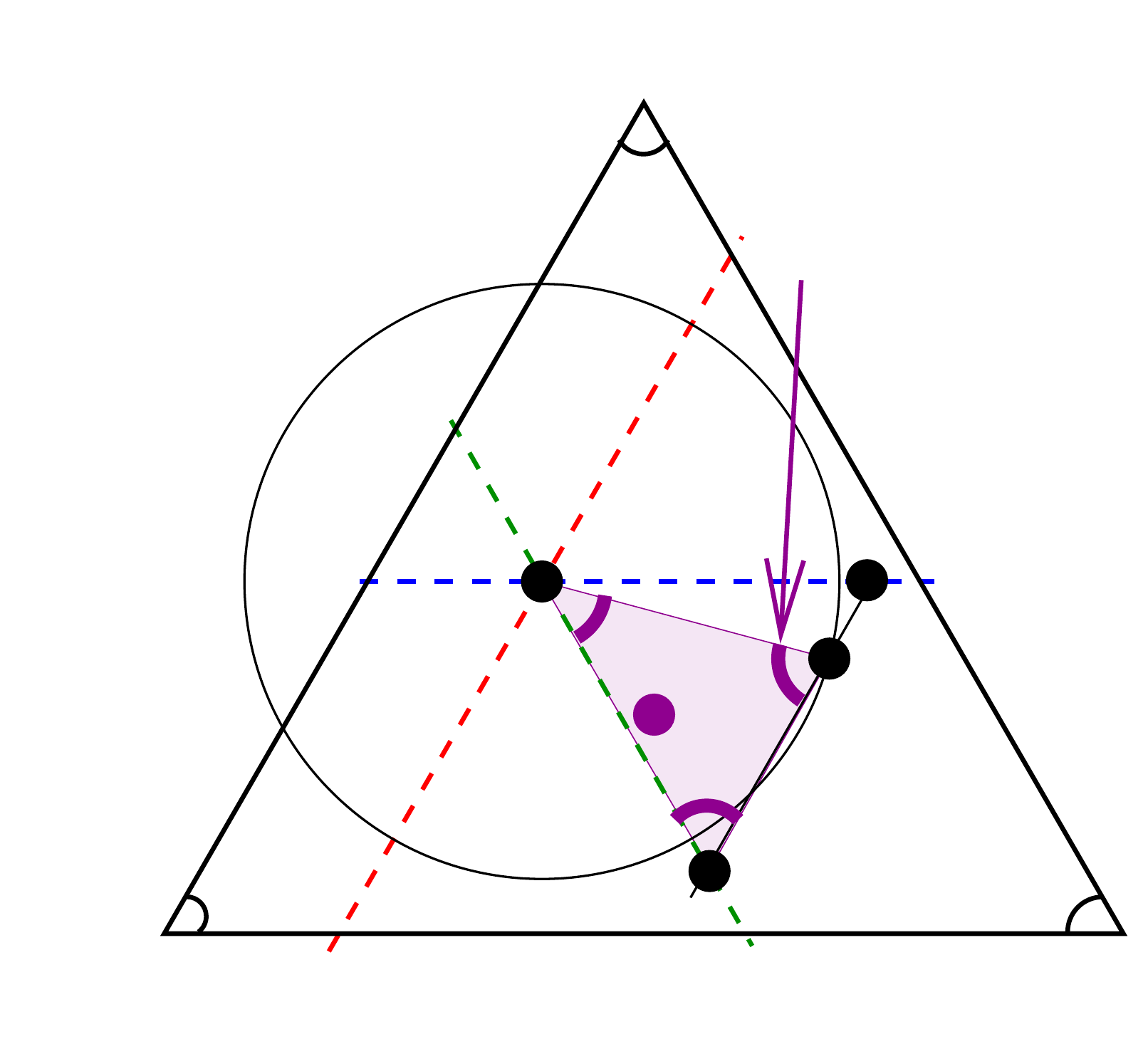\label{fig:connectivitya}}}\hspace{3cm}
\subfigure[]{\scalebox{0.3}{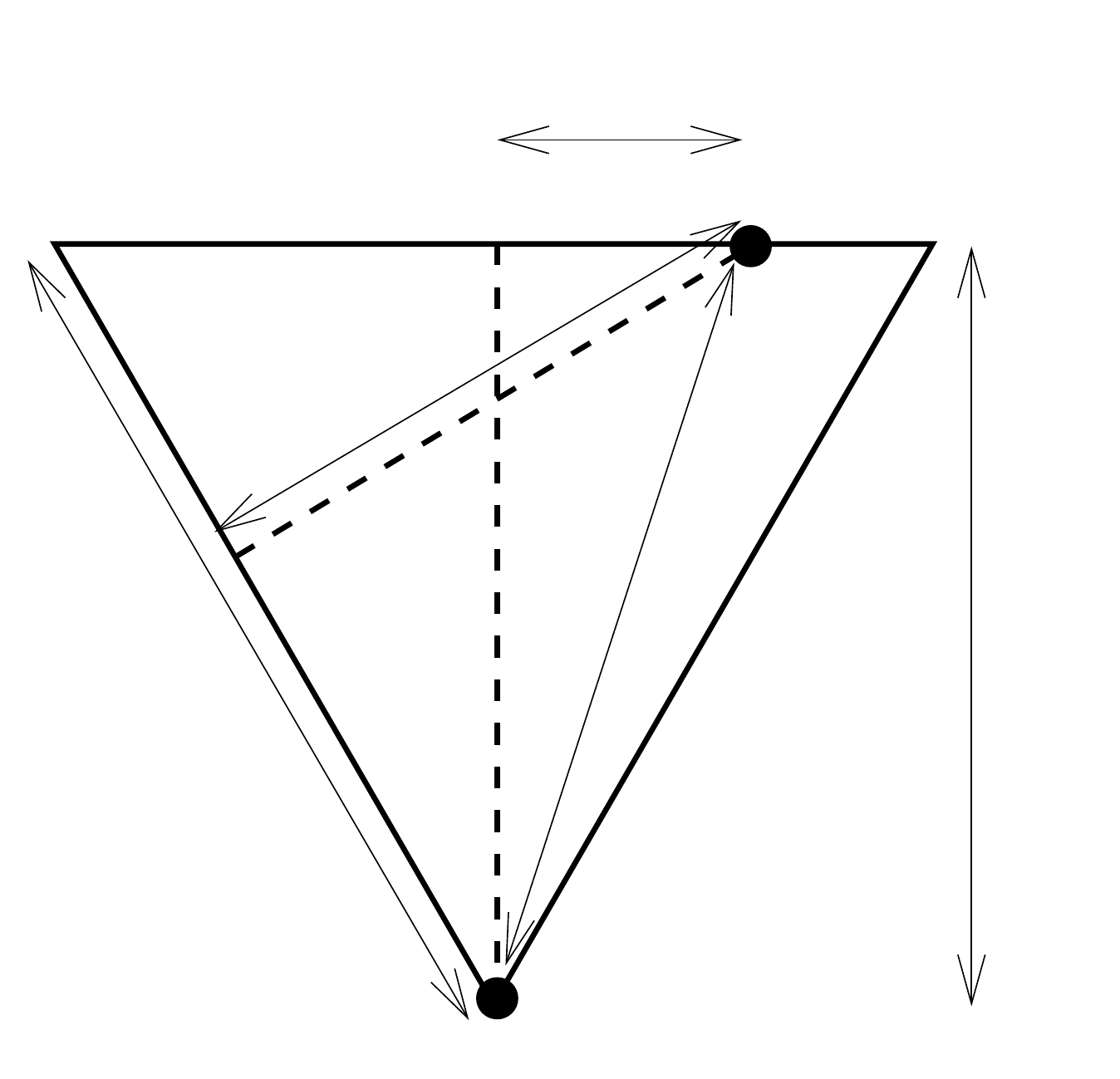\label{fig:connectivityb}}}
\caption{$|x'y| < |xy|$ and $\sum_{e \in P} |e \leq 2|xy||$}\label{fig:connectivity}
\end{center}\end{figure}

Using the notation of Figure \ref{fig:connectivitya}, we obtain that $|x'y|$ is maximum for $x'=x$, $x'=c_1$ or $x'=c_2$. We hence have to prove $|c_2y| \leq |xy|$. The case $|c_1y| \leq |xy|$ is symmetric.

In the purple triangle of Figure \ref{fig:connectivitya}, we have $\frac{|c_2y|}{sin(2\pi/3-\beta)}=\frac{|xy|}{sin(\pi/3)}$. So $\frac{|c_2y|}{|xy|}=\frac{sin(2\pi/3-\beta)}{sin(\pi/3)}$.
As $\beta \geq \pi/3$, we have $\frac{|c_2y|}{|xy|} \leq 1$.
So $|c_2y| \leq |xy|$.

Recall that we supposed that no two nodes in the network have the same coordinates compared to a given anchor. It means that $x'$ can neither be $c_2$ nor $c_1$\footnote{Because $d(x,A_2)=d(c_1,A_2)$ and $d(x,A_3)=d(c_2,A_3)$}. We hence have $|x'y| < \max(|xy|,|c_1y|,|c_2y|) \leq |xy| \leq r$.

So we can apply the induction on the edge $(x'y)$ which is in $E$ and strictly shorter than $(xy)$ and this proves that there is a path between $x$ and $y$ through $x'$. 

Let us now look at the stretch factor. We have $x <_1 x'$, $y>_1 x'$ and there is a path $y^1=x', y^2,\ldots, y^l=y$ in the graph $\widetilde G$ that satisfies $y^1~<_1 y^2 <_1 \ldots <_1 y^l=y$ by construction of the virtual edge. Because the coordinate with respect to $A_1$ increases(monotonically)  along the path we have that $| x_1-y^1_1| + | y^2_1-y^3_1| + \ldots + | y^{l-1}_1-y^l_1| = | x_1-y_1|$ (the subscript indicates the coordinate with respect to $A_1$).

If $x'\in A_i^y$ for $i\in\{2,3\}$, we have $y <_i x'$ and $y >_1 x'$. By induction, if such a $i$ exists, the rest of the path will be in $ A_i^y$ and the $i$th-coordinate decreases (monotonically) along the path. If not, $x'\in \bar A_1^y$ and $y >_1 y'$, $y<_2 x'$ and $y<_3 x'$ and both coordinates with respect to $A_2$ and $A_3$ decreases along the path. This proves that in all cases, we have that there exists $i\in\{2,3\}$ such that $ | x_i - y_i^1| + | y^{1}_i-y^2_i|+\ldots + | y^{l-1}_i-y^{l}_i|  = | x_i-y_i|$.

In summary, along the path $P$, for $z\in P$, $z_1$ increases from $x_1$ to $y_1$ and there is an $i\in\{2,3\}$ such that $z_i$ decreases from $x_i$ to $y_i$. The distance covered in $A_1^x$ to go from a node with $i^{th}$ coordinate $x_i$ to a node with $i^{th}$ coordinate $y_i$ is upper-bounded by $\frac{2}{\sqrt{3}}|x_i-y_i|$, see Lemma \ref{length}, since we suppose that the anchors form an equilateral triangle. 

From this we deduce that the length of the path $P$ is upper bounded by $\frac{2}{\sqrt{3}}|x_1-y_1|+\frac{2}{\sqrt{3}}|x_i-y_i|$, this longest path is obtained by moving along the path where the $i$th coordinate is constant and the first coordinate goes from $x_1$ to $y_1$ and then along the path where the $i$th coordinate goes from $x_i$ to $y_i$ and the first one is constant.

We now express $|xy|$ in terms of $|x_1-y_1|$ and $|x_i-y_i|$. By the configuration of the different triangles, c.f. Figure \ref{fig:connectivityb}, we have $|xy|^2=|x_1-y_1|^2 + (\frac{2}{\sqrt{3}}|x_i-y_i| - \frac{1}{\sqrt{3}}|x_1-y_1|)^2=\frac{4}{3}(|x_i-y_i|^2+|x_1-y_1|^2-|x_i-y_i||x_1-y_1|)$.

Hence the stretch factor $c$ verifies: $c^2=\frac{(\frac{2}{\sqrt{3}}|x_1-y_1|+\frac{2}{\sqrt{3}}|x_i-y_i|)^2}{\frac{4}{3}(|x_i-y_i|^2+|x_1-y_1|^2-|x_i-y_i||x_1-y_1|)}\leq 4$.

So we have $c\leq 2$ as claimed.
\end{proof}

\begin{corollary}
$\widetilde G\prime $ is a subgraph of $G^{\cal S}$ which is equal to the half-$\theta_6$ graph.
\end{corollary}
\begin{proof}
The definition of the half-$\theta_6$ graph gives the same graph as $G^{\cal S}$ when using the three total orders using the modified VRAC coordinates. From the previous lemma, we get that if there is an edge $(xy)\in \tilde E\prime$, with $y$ in a greedy region $A_k^x$, then $y$ is minimum according to $\tilde <_k$. Hence, $(xy)$ is an edge of $G^{\cal S}$. However, there are examples in which $\widetilde G\prime \neq G^{\cal S}$
\end{proof}

From this corollary, we immediately obtain that $\widetilde G\prime$ is planar. When $\tilde G\prime=G^{\cal S}$, we also deduce that it is geometric 2-spanner. However, in the general case, we can not deduce any information on its stretch, which is the object of the next theorem, which is implied by Lemma \ref{edgeFactor2}.

\begin{theorem}\label{th:2spanner}
Given a connected graph $G$, the graph $\widetilde G \prime$ is planar, and for any two nodes $x$ and $y$, if there is a path of length $\ell$ from $x$ to $y$ in $G$, there is one of length at most $2\ell$ in $\widetilde G \prime$.
\end{theorem}

However, notice that the previous theorem applies to $\widetilde G\prime$ which may contains virtual edges that are not edges of the Unit Disk Graph $G$, instead, a virtual edge represents a path in $G$. The next lemma says that such a path has length two. Notice further that the virtual edges are edges of the Unit Hexagonal Graph (c.f. \cite{review}). Finally, simulations in Section \ref{experiments} show that the virtual edges are rare and disappear as the node density increases.

\begin{lemma}\label{length}
A virtual link $(xy)\in \widetilde E \prime$ has length at most $2r/\sqrt{3}$ when the anchors form an equilateral triangle.
\end{lemma}

\begin{proof}
Without loss of generality, we can suppose that the link is oriented from $x$ to $y$ and with respect to minimizing coordinate $y_1$. The virtual links represents a path $x,z,\dots,y$ where $z$ is inside the communication range of $x$, so $d(x,z) \leq r$. 
We know by construction that considering the first coordinates, we have $y_1 \leq z_1$. It means that $y$ is in the triangle $T_1$ delimited by the three lines $\{u|u_1=z_1\}$, $\{v|v_2=x_2\}$ and $\{w|w_3=x_3\}$ as depicted in Figure \ref{equilateral}. The furthest points of this triangle are the two summits other than $x$. This triangle is equilateral and the edges are of length $2r/\sqrt{3}$, so $|xy| \leq 2/\sqrt{3}r$.
\end{proof}

From this lemma, we obtain that Algorithm \ref{algo} constructs correctly $\tilde G\prime$.

\begin{algorithm}
\begin{algorithmic}
\STATE Input: A Unit Disk Graph $G$.
\STATE Output: $\tilde G  \prime$.
\FOR{all $x\in V$}
\FOR{$k\in {1,2,3}$, $y_k = \min_k\{y \in A_k^x, |y_kx|<r$\}}
	\STATE $x$ broadcast "activate $(xy_k)$"
\STATE $y_k'=\min_k\{y \in A_k^x, |y_ky|<r\}$
	\IF{ $y_k \neq y'_k$}
		\STATE $y_k$ broadcast "disable $(xy_k)$ and activate $(xy'_k)$"
	\ENDIF
\ENDFOR
\ENDFOR
\end{algorithmic}
\caption{Distributed construction of $\tilde G\prime$.}\label{algo}
\end{algorithm}

\begin{corollary}\label{subgraph}
Given a connected graph $G$, if the density is high enough, $\widetilde G \prime$ is a planar 2-spanner of $G$. 
\end{corollary}

\begin{proof}
When the density increases, the number of virtual edges decreases, and if the density is high enough, with high probability, there remains none of them. In this case, $\widetilde G \prime$ is a subgraph of $G$, and the result follows from Theorem \ref{th:2spanner}. 
\end{proof}

To summarize, $\widetilde G\prime$ is a planar graph such that 1) the length of a shortest path in this graph is at most twice the length of a shortest path in the communication graph, and 2) virtual edges correspond to a path of length two in the communication graph. Furthermore, to construct $\widetilde G\prime$ (Algorithm \ref{algo}), the communication complexity in terms of bits at a node $x$ is as follow: each node broadcast once its Id and coordinates, then, each node broadcast the Id of its three neighbors minimizing the orders, and finally, each of this neighbors may send to $x$ an other node Id which is the extremity of a virtual edges starting from $x$. Hence, each node induces an exchange of at most 6 nodes (excluding the neighborhood discovery), which gives a total of at most $6n$ nodes' ids that are broadcast. The computational complexity is $O(n\Delta)$, with $\Delta$ the maximum degree of the communication graph: each node $x$ computes which of its neighbors minimizes each order, plus, $x$ requires each of the selected neighbors $y_1,y_2,y_3$ to verify if there should be a virtual edge, which also consists in computing a minimum of a set of at most $\Delta$ elements.

Now, if instead of using VRAC coordinates, we use the Euclidean coordinates, each node $y$ can computes on its own if it is minimum for one of the three orders for one of its neighbors. It means that we can avoid the statement "broadcast "activate $(xy_k)$"" in Algorithm \ref{algo}. The modified version of Algorithm \ref{algo} needs the broadcasts of at most $3n$ nodes identifiers that can all be performed in a single round of communication. In case there are no virtual edges, the resulting graph is a 2-spanner of the Unit Disk Graph, and in this case no messages are exchanged. It answers the open question 22 of \cite{review} under the hypothesis that the density id high enough. Recall that $\tilde G\prime$ is always a subgraph of a Unit Hexagonal Graph (c.f. \cite{review}), and hence a planar spanner for these graphs.

\section{A local routing algorithm}
We  now propose a local routing algorithm. This algorithm has two modes depending on if the destination is in a greedy region of the sender or not.

\begin{lemma}\label{greedylemma} Recall that we suppose the triangle $\widehat{A_1A_2A_3}$ equilateral. Let $x$ be a node with a message for a destination $y$. If $y$ belongs to a greedy region of $x$, and that $x$ has an out-neighbor in this greedy region, the algorithm proceeds as follow:
\begin{itemize}
\item {\bf (Data delivery)} $|xy| \leq r$ in which case $x$ transmits the data directly to $y$.
\item {\bf (Greedy routing)} $|xy| > r$ and $x$ transmits to its neighbor $x'$ that belongs to the same greedy region as $y$.  We have $|xy|\ge |x'y|$. 
\end{itemize}
\end{lemma}
\begin{proof} If $x$ transmits directly to $y$ there is nothing to prove. For the other case, let $x'$ be the node that receives the message. We note the vector $xy=a \text{e}^{i\alpha_1}$ and $xx' = b\text{e}^{i\alpha_2}$ ($| xy| = a$ and $a\ge b$). The vector $x' y=a\text{e}^{i\alpha_1}-b\text{e}^{i\alpha_2}$, $| x'y| = a^2+b^2-2ab\cos(\alpha_1-\alpha_2)$ and, $0\le \alpha_1,\alpha_2\le \alpha$ because $x'$ and $y$ belong to the same greedy region. Then, $-\alpha\le \alpha_1-\alpha_2\le\alpha$ and $\cos(\alpha_1-\alpha_2)>1/2$. This is sufficient to conclude that $| x' y| \le | xy|=a$.
\end{proof}

\begin{lemma}\label{structure}
We assume that the anchors form an equilateral triangle. Let $x$ be a node that wants to send a message to a destination $y$ using the graph $\widetilde G \prime$ or direct transmission. We assume that $y$ does not belong to a greedy region of $x$. Without loss of generality, $y$ is in $\bar A_2^x$.
We further assume that $x$ has three out-neighbors and so do any node in the equilateral triangle $T$ with base the segment parallel to $(A_2A_3)$, centered at $x$, of length $4/\sqrt{3}r$ and with the other summit in $\bar A_2^x$ (c.f. Figure \ref{equilateral}). Under those hypothesis, we have two paths (without consideration on the orientation of the edges) $P_1=x,u_0,P_1^0,u_1,...,u_{k-1},P_1^{k-1},u_k,z$ and $P_2=x,v_0,P_2^0,v_1,...,v_{l-1},P_2^{l-1},v_l,z$, without loss of generality $u_0 >_1 v_0$ (when $P_2 \neq x,z$, as if  $P_2 = x,z$, there are no $v_0$), and we have:  
\begin{itemize}
\item $k-1 \leq l \leq k$.
\item the $P_i^j$, for $i \in \{1,2\}$ and $j\leq k$ are monotone paths with respect to $>_1$, potentially of length 0.
\item $\forall u \in P_1\setminus\{z\},\ u \in A_2^x$.
\item $\forall v \in P_2\setminus\{z\},\ v \in A_3^x$.
\item $\forall 0 \leq i \leq l$, there is an oriented edge from $u_i$ to $v_i$ and, $\forall u \in \{u_i,\ P_1^i\}$, $u >_1 v_i$.
\item $\forall 0 \leq i \leq l-1$, there is an oriented edge from $v_i$ to $u_{i+1}$ and, $\forall v \in \{v_i,\ P_2^i\}$, $v >_1 v_{i+1}$.
\item $z$ is in $\bar A_2^x$.
\end{itemize}

Given these two paths, either a node from $\{x,u_0,...u_k,v_0,...v_l,z\}$ has $y$ within its communication range, or $|zy| < |xy|$.
\end{lemma}

\begin{figure}[h]
\begin{center}
\scalebox{0.4}{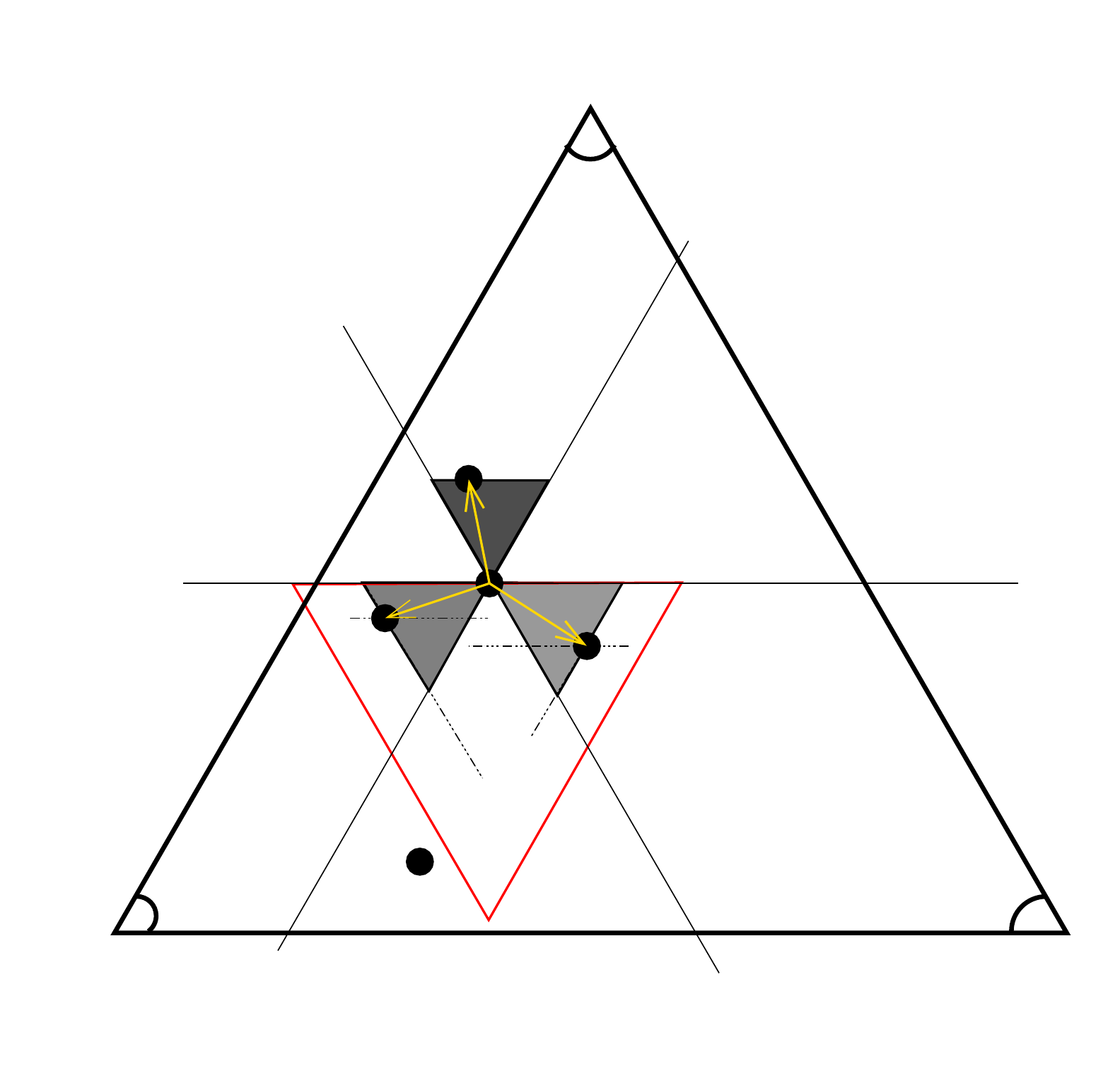} \hspace{2cm}
\scalebox{0.4}{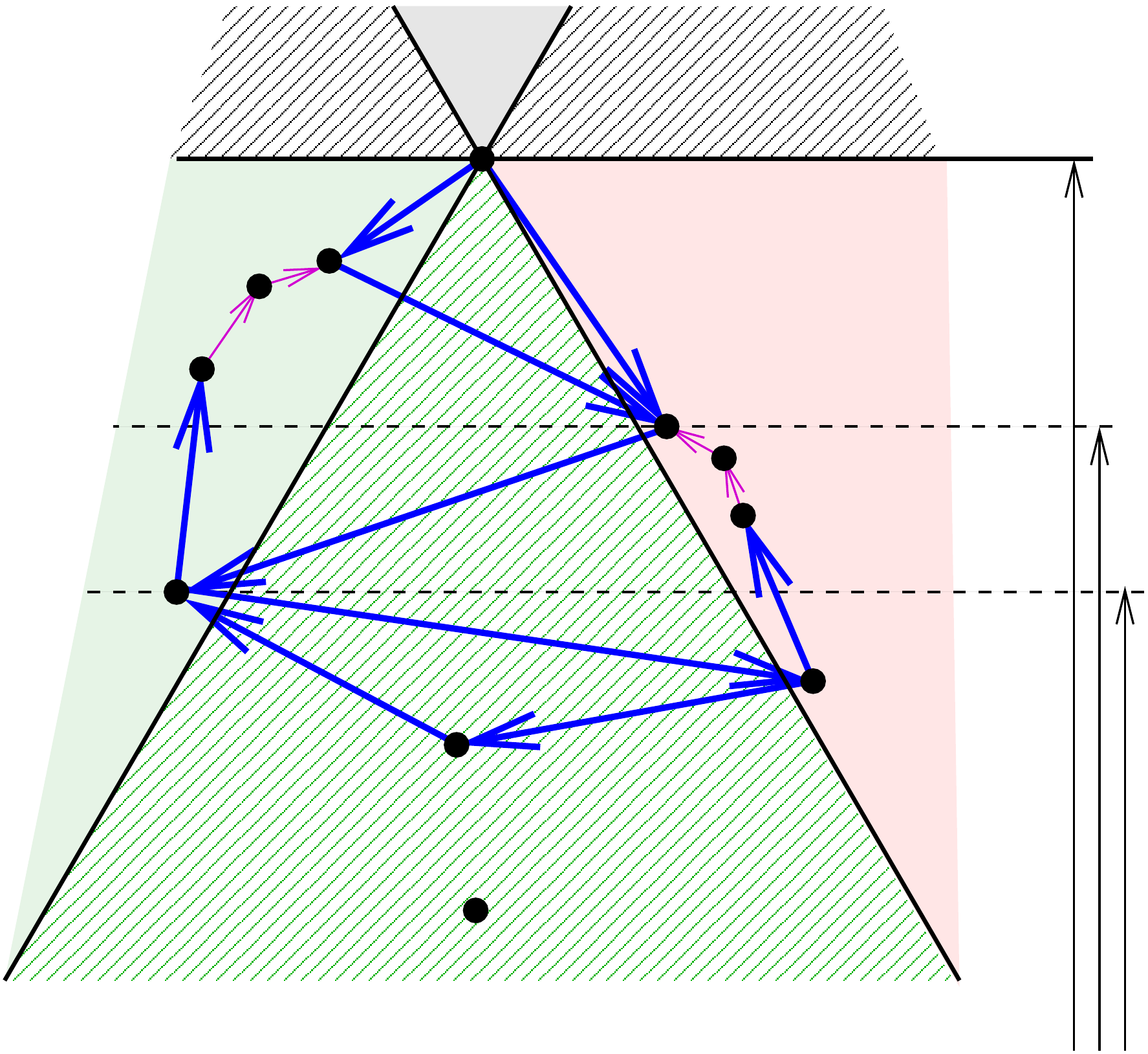}
\caption{On the left side, the three greedy regions associated with $x$ and the edges to the three nodes $u=min_2\{z\mid x\widetilde<_2~z\},\ v=min_3\{z\mid x\widetilde<_3~z\}\ and,\ w=min_1\{z\mid x\widetilde<_1~z\}$. On the right, a path leading to $z$.}\label{equilateral}
\end{center}
\end{figure}

\begin{proof}
By hypothesis, $x$ has a neighbor (in $\widetilde G \prime$)  $u_0$ in $A_2^x$ and a neighbor $v_0$ in $A_3^x$, without loss of generality $u_0 >_1 v_0$. By Lemma \ref{length}, which bounds the length of an edge, $u_0$ is in $T$, so $u_0$ has three out-neighbors. In particular it has an out-neighbor $u'$ in $A_3^{u_0}$. $u'$ can not be in $A_3^{u_0} \cap A_2^x$ as otherwise we would have a virtual link $(xu')$ instead of the link $(xu_0)$. 

If $u'$ is in $\bar A_2^x$, if we can prove that $(xu')\in \widetilde E\prime$ setting $z=u'$, the theorem would be verified with $P_1=x,u_0,z$ and $P_2=x,z$. But, we may not have an edge between $x$ and $u'$, however we prove that $u'$ can be replaced by a $z"$ satisfying all the desired properties. $u'$ is in $T$, so it has an out-neighbor $u"$ in $A_1^{u'}$. If $u" \not = x$, $u"$ has an out-neighbor in the greedy region oriented towards $A_2^{u"}$ that in turn has a neighbor in its greedy region oriented towards $A_2$, ... By planarity, since we cannot cross the edge $u_0z$ this path has to go through $v_0$. So we have a path from $u"$ to $u_0$. Similarly, $u"$ has an out-neighbor $u"$ in $A_1^{u"}$ which is either $x$ or has a path leading to $u_0$. Ultimately, we will have a node $z'$ with out-neighbor $x$ and a path to $u_0$. Using the same argument and considering neighbors in $A_2^{z'}$, we have that $z'$ has $u_0$ has neighbor in $A_2^{z'}$, or it has $w$ whose out-neighbor in $A_1^w$ is $x$, and that has a path leading to $u_0$ in $A_2^{w}$. Going along this path, we reach $z"$ whose out-neighbor in $A_1^{z"}$ is $x$, and  out-neighbor in $A_2^{z"}$ is $u_0$. Hence the theorem is verified for  $P_1=x,u_0,z"$ and $P_2=x,z"$

We now suppose that $u'$ is in $A_3^x$ and we want to prove $u'=v_0$. By contradiction suppose that $u' \not=v_0$. Recall that we have $u' <_1 u_0$, by construction of $\widetilde G\prime$ as well as $u' <3 v_0$. By planarity, we have $v_0 >_1 u'$.

The node $u'$ has a neighbor in its greedy region $A_1^{u'}$. This neighbor must be $v_0$. Indeed, if not, using $u' <3 v_0$ and using the same argument as before, this neighbor $u"$ has a neighbor in the greedy region $A_2^{u"}$ that in turn has a neighbor in its greedy region oriented towards $A_2$, ... But, since we cannot cross the edge $u_0u'$ by planarity, there is an infinity of such nodes which is impossible. So $u'$ has $v_0$ has neighbor in its greedy region $A_1^{u'}$. But then $v_0$ has a neighbor in its greedy region $A_2^{v_0}$. Looking for a sequence of neighbors in the greedy region oriented towards $A_2$, will ultimately lead to cross edge $(u_0u')$ which contradict the planarity of $\widetilde G \prime$.

So $u' = v_0$

\ \\

Since $v_0$ is in $T$, it has a neighbor $v'$ in its greedy region oriented towards $A_2$. $v'$ cannot be in $A_3^{u_0}$, the same region as $v_0$ because if it was true, by construction of $\widetilde G\prime$, there would be an edge between $u_0$ and $v'$ instead of the edge $(u_0v_0)$.

\ \\

First suppose that $v' \in \bar A_2^x$.
We would like to set $z=v'$, but, as for the very first case considered in the proof, it may not be appropriate. Indeed, there may not be an edge $(zv_0)$. So we proceed as follow: we consider the out-neighbor $v"$ of $v'$ in $A_1^{v'}$. If $v"=v_0$, then we are done by setting $z=v'$. Else, we can prove that $v"$ has $u_0$ as out-neighbor in $A_3^{v"}$. We proceed recursively until we reach $v"'$ which has $u_0$ as out-neighbor in $A_3^{v"'}$ and $v_0$ as out-neighbor in $A_1^{v"'}$ (Notice that $u_0$ can not have $v"'$ as out-neighbor in $A_3^{u_0}$ as it already has $v_0$). We set $z=v"'$ and the theorem is proved.

\ \\

We now study the case $v' \in A_2^x$.
Let call $v"$ the out-neighbor of $v'$ that is in the greedy region oriented towards $A_1$. If $v"=v_0$, all is fine. Else, we must have $v">_0 u_0$ (otherwise the neighbors directed towards $A_3$ will cross the edge $v_0v'$). Considering the following out-neighbors in the greedy region oriented towards $A_1$, we will reach a node $v"'$ whose out-neighbor in $A_1^{v"'}$ is $u_0$, or which is the out-neighbor of $x$ in $A_2^{x}$. We call this path from $v"$ to $v"'$, $P_1^0$ and we set $v'=v_1$.

\ \\

We continue similarly by considering the out-neighbor of $v_1$ in $A_2^{v_1}$.

\ \\

To see that this process terminates, notice that the first coordinate of the vertices on the paths we construct decrease as we get closer to the line $(A_2A_3)$, plus the distance between a node in $A_2^x$ and a node in $A_3^x$ is lower bounded by a bound which increases when we get closer to the line $(A_2A_3)$, plus the length of an edge is upper bounded by Lemma \ref{length}. Hence the process has to finish on a vertex $z$ in $\bar A_2^x$.

\ \\

We now prove the second part of the Lemma.
The maximum length of a virtual edge is $2/\sqrt{3}r$, hence we have that the polygon formed by $x,u_0,...u_k,z,v_l,...v_0$ is composed of triangles with two side of length at most $2/\sqrt{3}$. So it is covered by the union of the disk of radius $r$ centered on $\{x,u_0,...u_k,v_0,...v_l,z\}$.
If $y$ is inside the polygon formed by $x,u_0,...u_k,z,v_l,...v_0$, it is within the communication radius of one of these vertices.
If not, $y$ is bellow the polygon $x,u_0,...u_k,z,v_l,...v_0$ and out of range from both $v_l$ and $u_k$ in particular, and then $|zy| < |xy|$ .
\end{proof}

\begin{theorem}[Zig-Zag : an extended greedy routing]\label{thm:zigzag}

We assume that the anchors form an equilateral triangle. Let $x$ be a node that needs to transmit a message to a destination node $y$. If any node at distance less than $4/\sqrt{3}r$ of $x$ has three out-going neighbors, then the following strategy delivers the data either to $y$ or to a node $z$ closer to $y$ than $x$.

\begin{itemize}
\item If $y$ is in the communication range of $x$, $x$ sends the message to $y$.
\item If $y$ is in a greedy region of $x$, $x$ sends the message to its out-neighbor which is in the same greedy region. 
\item Otherwise, use the restricted greedy routing process starting at $x$. Wlog $y \in \bar A_2^x$.
\begin{itemize}
\item $x$ sends the message to its out-neighbor in $A_2^x \cup A_3^x$ which has the highest first coordinate.
\item A node $u\in A_2^x$ sends the message to $y$ if possible or to its out-neighbor $v$ verifying $u>_1 v$ and $v>_3 x$. If $v>_3 x$ and $v >_2 x$ (i.e. $u \in \bar A_2^x$), end the restricted greedy routing process.
\item A node $v\in A_3^x$ sends the message to $y$ if possible or to its out-neighbor $u$ verifying $v>_1 u$ and $u>_2 x$. If $u>_3 x$ and $u >_2 x$ (i.e. $u \in \bar A_2^x$), end the restricted greedy routing process.
\end{itemize}
\end{itemize}
\end{theorem}

\begin{proof} 
By applying Lemmas \ref{greedylemma} and \ref{structure} we see that the routing strategy leads to $y$ or to a node that is closer to $y$ than $x$ ($z$ in Lemma \ref{structure}). 
Indeed, Lemma \ref{structure} ensures that a restricted greedy routing process starting at $x$ follows the path $u_0,v_0,u_1,...,z$ using the same notations.\end{proof}

\section{Simulations}\label{experiments}
We implemented both the planarization algorithm and Zig-Zag, the routing algorithm, and we present bellow the results of the simulations. We considered a network composed of 300 sensors spread in a square area $[0;1]\times [0;1]$ with three anchors at position $(0.5,3.5)$, $(-\frac{5}{\sqrt 3} + 0.5,-1.5)$ and $(\frac{5}{\sqrt 3} + 0.5,-1.5)$. We considered a communication radius $r$ for the sensors which ranges from $0.11$ to $0.225$. For each value of the communication radius, we performed the average over 1000 networks and successfully routed messages. Notice that the value are plotted with respect to the average degree of the nodes in the UDG, as in $\tilde G\prime$, the average degree is upper bounded by six since it is planar.

We first plot the number of virtual edges in Figure \ref{simedges}. The simulations indicate that the number of virtual edges tends to 0 when the node density increase. Indeed, the average number of virtual edges decreases from 1.6 when the network is sparse, to 0.03 when it is dense.

Then, in Figure \ref{simstretch}, we plot the average stretch of the path computed in $\tilde G \prime$ by Zig-Zag, where the stretch of a path is the length of the computed path divided by the euclidean distance between the source and the destination. We observe that we have a stretch which is between 1.3 and 1.4 which is better than the theoretical stretch factor of 2, and we compare it to the stretch of the greedy algorithm (in $\tilde G\prime$), which is slightly better and of approximately 1.3. This apparent gain of efficiency is to be mitigated by the results shown on Figure \ref{success} which indicates the success rate of the greedy algorithm and of Zig-Zag. The success rate of Zig-Zag tends to 100\% as the density increases, which backen the theoretical results of Theorem \ref{thm:zigzag}. Indeed, when the density increases, the hypothesis of Theorem \ref{thm:zigzag} are verified, whereas they may not be verified at some nodes when it is low, thus reducing the success rate. The success rate of the greedy algorithm also increases with the node density, however, it remains bellow 80\%.

\begin{figure}[h]
\begin{center}
\subfigure[Average number of virtual edges]{\scalebox{0.7}{\includegraphics{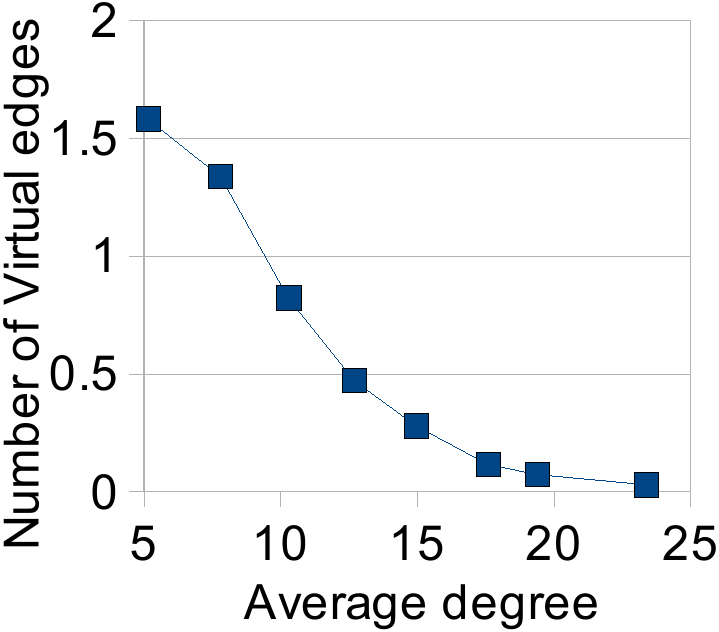}\label{simedges}}}\hspace{1cm}
\subfigure[Average path stretch]{\scalebox{0.7}{\includegraphics{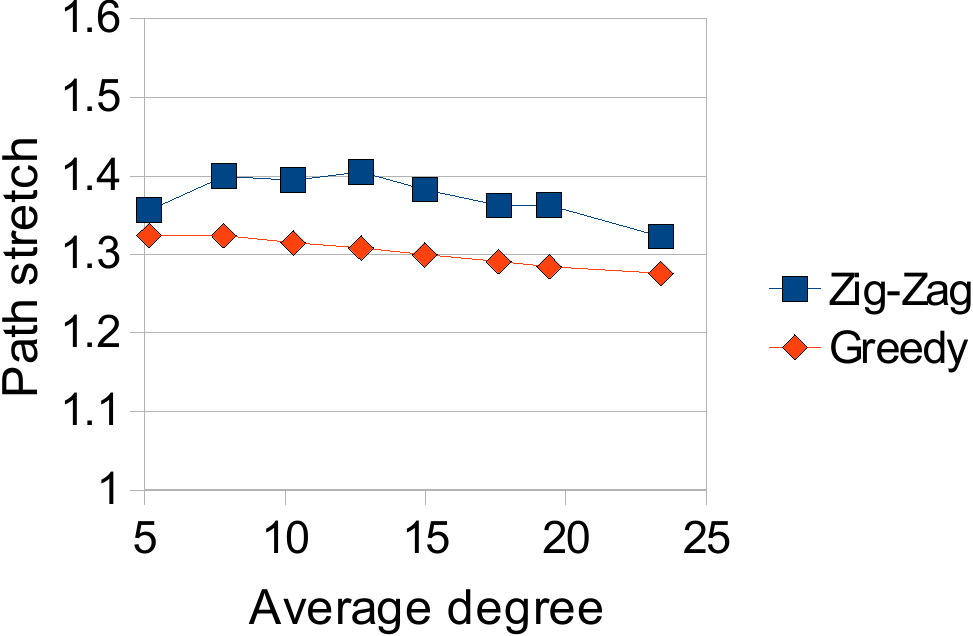}\label{simstretch}}}\hspace{1cm}
\subfigure[Success rates]{\scalebox{0.7}{\includegraphics{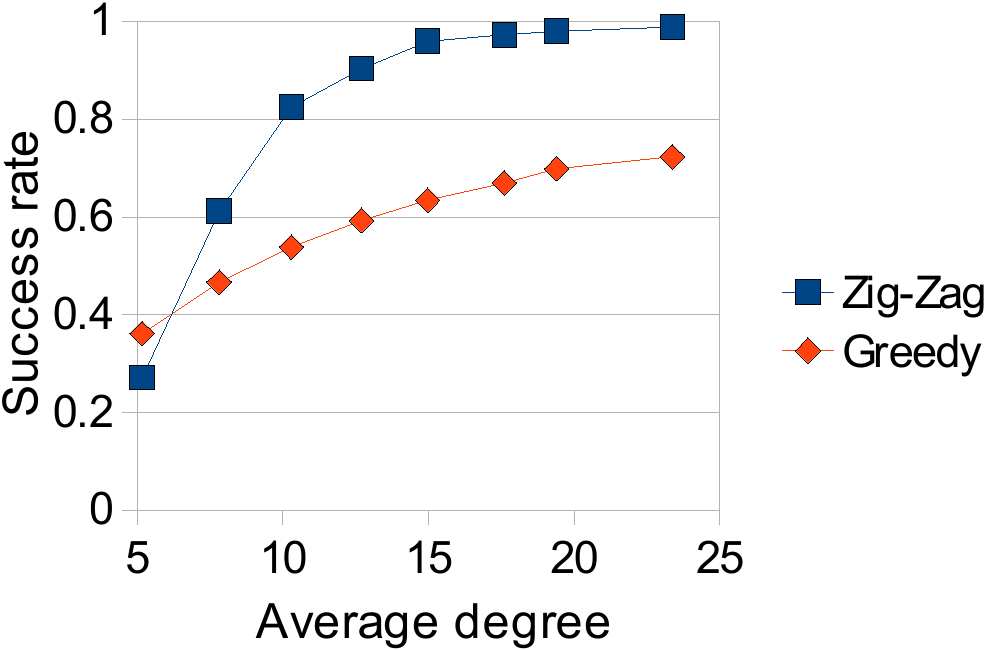}\label{success}}}
\caption{}\label{sim}
\end{center}
\end{figure}
\bibliographystyle{IEEEtran}
\bibliography{Localisation}
\end{document}